\pdfoutput=1
\PassOptionsToPackage{table}{xcolor}
\documentclass[sigconf]{acmart}

\usepackage{amsmath}
\usepackage{booktabs}
\usepackage{multirow}
\usepackage{graphicx}
\usepackage{xcolor}
\usepackage{enumitem}
\usepackage{balance}

\newcommand{\bestcell}[1]{\cellcolor{gray!15}\textbf{#1}}

\AtBeginDocument{}
\theoremstyle{definition}
\newtheorem{defn}{Definition}

\copyrightyear{2026}
\acmYear{2026}
\setcopyright{cc}
\setcctype{by}
\acmConference[CIKM '26]{Proceedings of the 35th ACM International Conference on Information and Knowledge Management}{November 7--11, 2026}{Rome, Italy.}
\acmBooktitle{Proceedings of the 35th ACM International Conference on Information and Knowledge Management (CIKM '26), November 7--11, 2026, Rome, Italy}
\acmISBN{979-8-4007-2539-5/2026/11}
\acmDOI{10.1145/3799682.3841132}
\begin{document}

\title{The Recall Ceiling of LLM Recommendation Reranking}

\author{Zhaohui Wang}
\orcid{0009-0006-1187-1903}
\affiliation{%
  \institution{University of Southern California}
  \department{Viterbi School of Engineering}
  \city{Los Angeles}
  \state{California}
  \country{USA}}
\email{zwang000@usc.edu}

\renewcommand{\shortauthors}{Zhaohui Wang}

\begin{abstract}
Some LLM-based recommendation rerankers are evaluated under an \emph{oracle} protocol that guarantees the ground-truth item is present in the scored set, either by injecting it into the candidate list or by scoring it against sampled negatives.
We show that this protocol overestimates realistic performance by 92--95\% in NDCG@10 across three primary Amazon datasets, so oracle results do not transfer to deployment.
The root cause is a \emph{recall ceiling}: realistic retrieval covers only 2--19\% of relevant items at $K{=}100$ across eight datasets in three domains (Amazon products, MovieLens movies, MIND news), placing a deterministic upper bound on any closed-candidate reranker's top-$k$ NDCG (Theorem~\ref{thm:ceiling}, $\mathbb{E}[\text{NDCG@}k] \leq \text{Recall@}|W_\pi|$ under LOO evaluation, where $W_\pi$ is the reranker's own candidate window).

Under realistic retrieval, every optimisation strategy we tested fails to produce a statistically significant improvement over the CF baseline on our primary Amazon datasets: prompt engineering, model scaling across a 168$\times$ parameter range, sequential models, three supervised neural rerankers (LambdaMART, RankNet, MLP) under a disjoint train/eval user split, a LoRA-fine-tuned LLaMA-3.2-3B reranker, and a hybrid CF+BM25+Dense retrieval with reciprocal-rank fusion.
Adding text-aware retrieval raises recall on Beauty but does not lift end-to-end NDCG.
Giving the LLM the upstream CF ranks and scores directly---the fairest setting we can construct---narrows but does not close the gap to CF, and does so by making the LLM adhere more closely to the CF order rather than by adding information; an LLM$+$CF fusion tuned on the evaluation users themselves recovers CF's performance and no more.
We propose the Recall-Aware Evaluation Protocol (RAEP) as a diagnostic: classify the regime by retrieval recall, then evaluate reranking only where the ceiling permits differentiation.
In the retrieval regimes we can measure (Recall@100 of 2--19\%), the path to effective LLM recommendation runs through retrieval improvement rather than reranker sophistication; we do not claim this ordering holds at the higher recall, richer feature sets, and online feedback available to production systems.

\end{abstract}



\maketitle

\begin{figure*}[t]
  \centering
  \includegraphics[width=0.70\textwidth]{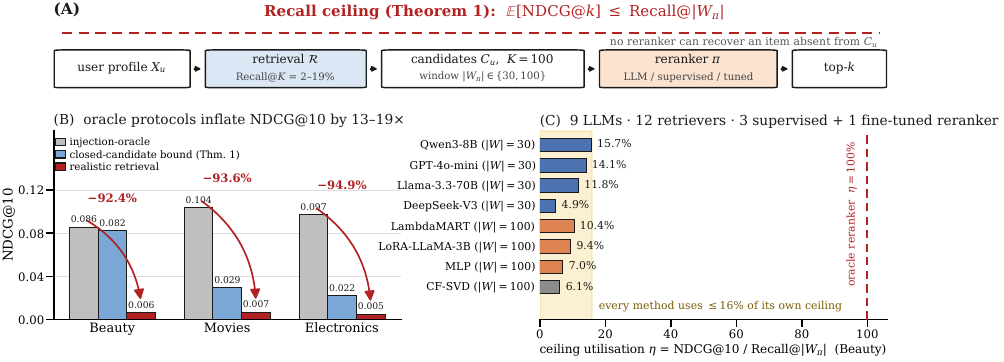}
  \caption{\textbf{Overview.} (A) The retrieve-rerank pipeline; Theorem~\ref{thm:ceiling} caps any closed-candidate reranker's top-$k$ NDCG by the recall of its own window $W_\pi$. (B) Injection-oracle NDCG@10, the closed-candidate bound, and realistic retrieval on the three primary datasets. (C) On Beauty, every reranker we test---zero-shot LLMs, supervised models, and a LoRA-fine-tuned LLaMA-3.2-3B---exploits under $16\%$ of its own ceiling, and none beats CF significantly. Generated by \texttt{scripts/make\_fig\_overview.py} from the released result JSONs.}\label{fig:overview}
\end{figure*}

\section{Introduction}\label{sec:introduction}

Recent work applies LLMs to recommendation through zero-shot ranking, LLM-augmented item text, instruction tuning, and generative retrieval~\cite{hou2024large,hanjia2024llmrec,bao2023tallrec,li2023gpt4rec}, borrowing prompting techniques from the wider LLM literature~\cite{wei2022chain,zhou2022large}. Some of these evaluations report gains under a protocol in which the held-out relevant item is guaranteed to sit in the scored set: an \emph{oracle evaluation protocol} that injects the ground truth into the candidate list before reranking~\cite[\S3.1]{hou2024large}, or a sampled-negative evaluation that scores the positive against a fixed pool of non-interacted items~\cite{hanjia2024llmrec}. We audit which published protocols do and do not have this property in \S\ref{sec:related}, and we are explicit that several prominent works---TALLRec~\cite{bao2023tallrec}, which classifies a single target item, and GPT4Rec~\cite{li2023gpt4rec}, which retrieves from the full corpus---are \emph{not} subject to this critique. In production no injection occurs: candidates come from upstream retrieval, and what the reranker receives is far worse than the oracle implies. On our three primary Amazon datasets, the LLM reranker's headline NDCG@10 drops by \textbf{92--95\%} between the two regimes. This paper asks where that drop comes from, and what it implies for the design of LLM-based recommendation.

\textbf{A mechanism.} A reranker cannot rank an item it does not see. Theorem~\ref{thm:ceiling} (\S\ref{sec:theory}) formalises this as a tight combinatorial bound: for any closed-candidate reranker, $\mathbb{E}[\text{NDCG@}k] \leq \text{Recall@}|W_\pi|$ under LOO evaluation---where $W_\pi$ is the window of candidates the reranker actually sees, which need not be all $K$---with a multi-positive specialisation that bounds the metric by the recovered-set IDCG ratio. At the realistic Recall@100 we measure across eight datasets (2--19\%), this collapses the achievable NDCG@10 to at most $0.02$--$0.19$. The ceiling is binding before any reranker decision is made.

\textbf{A double failure.} The ceiling explains why headline numbers shrink, but not why \emph{no} reranker we test significantly outperforms a $\$0.00$ CF baseline on our primary Amazon datasets. To separate the two effects, we measure the \emph{ceiling-utilisation ratio} $\eta$: empirical NDCG@10 divided by the closed-candidate upper bound from Theorem~\ref{thm:ceiling}, computed against each reranker's own candidate window (\S\ref{sec:tightness}). On Beauty, $\eta$ stays below $16\%$ for every method we test: $4.9$--$15.7\%$ for the zero-shot LLMs, $7.0$--$10.4\%$ for the supervised rerankers and the fine-tuned LLM. Recall sets a low ceiling; the reranker exploits only a small fraction of even that ceiling. Investment in better prompts, larger models, and richer training signals is investment in $\eta$, and $\eta$ is empirically small under realistic retrieval.

\textbf{A falsification.} We test this thesis under maximal stress along five dimensions: model capacity (4B to 671B, plus the April-2026 DeepSeek-V4 reasoning family with ``thinking'' mode), training signal (zero-shot, three prompt strategies, three supervised rerankers, a LoRA-fine-tuned LLaMA-3.2-3B), retrieval quality (twelve methods including a hyper-parameter-tuned LightGCN, multi-source RRF, and a hybrid CF$+$BM25$+$Dense retriever), evaluation regime (LOO, multi-positive last-5, closed-catalog robustness), and candidate budget (a $K$-sweep from 10 to 200 plus an adaptive-$K$ per-user strategy). No setting reaches statistical significance over CF on Beauty/Movies/Electronics; two of the largest models (Llama-3.3-70B and DeepSeek-V3-671B) and three V4 reasoning configurations trend \emph{below} CF on Movies, with V4-Pro-think significantly worse ($-43\%$, $p{=}.004$) at ${\sim}37\times$ the no-think cost. On MovieLens-25M (17.4\% recall, $6\times$ denser than Amazon) no LLM beats CF; on MIND News LLMs trend $+17$--$31\%$ but the 88\% zero-recall rate keeps the comparison non-significant.

\textbf{What the formalisation adds.} That ``retrieval bounds reranking'' is intuitive to practitioners is an objection we agree with: Theorem~\ref{thm:ceiling} is not a surprising statement, and its value is quantitative rather than conceptual (\S\ref{sec:dpi}). It fixes the numerical ceiling a measured NDCG must be read against, identifies the reranker's own \emph{window} rather than the retrieval budget $K$ as the binding quantity, and yields $\eta$. The intuition alone predicts neither the size of the oracle--realistic gap nor how small $\eta$ turns out to be.

\textbf{An honest correction.} An earlier version reported $p{<}0.0001$ supervised gains for LambdaMART above 8.8\% recall; that recall-threshold effect was a training--evaluation user overlap. Under the corrected disjoint-user protocol (\S\ref{subsec:neural_rerankers}) no supervised reranker beats CF at any tested recall level.

\subsection{Contributions}
\begin{itemize}[nosep,leftmargin=*]
  \item A formal \emph{recall ceiling} (Theorem~\ref{thm:ceiling}) with both LOO and multi-positive specialisations: $\mathbb{E}[\text{NDCG@}k] \leq \text{Recall@}|W_\pi|$ for any closed-candidate reranker (\S\ref{sec:theory}).
  \item Separation of the recall failure from the utilisation failure via $\eta$, scored against each reranker's own candidate window. On Beauty, the one primary dataset with enough retrieved positives to estimate it, $\eta$ stays below $16\%$ across nine LLMs and three supervised rerankers (\S\ref{sec:results}--\S\ref{sec:falsification}).
  \item A five-dimension stress test---model capacity, training signal, retrieval quality, evaluation regime, candidate budget---in which no configuration significantly beats CF, and the LoRA fine-tuned LLM is significantly \emph{worse} on Movies ($-39.5\%$, $p{=}.026$); plus a 92--95\% oracle/realistic gap validated across eight datasets and three domains.
  \item Identification and correction of a training-evaluation leak in prior supervised-reranker results, and the Recall-Aware Evaluation Protocol (RAEP) with adaptive-$K$ as a minimum standard for future work (\S\ref{sec:implications}).
\end{itemize}
\section{Related Work}\label{sec:related}

\paragraph{LLMs for Recommendation.}
The application of LLMs to recommendation spans zero-shot ranking~\cite{hou2024large}, conversational recommendation~\cite{zhang2023recommendation}, LLM-augmented item text~\cite{hanjia2024llmrec}, instruction-tuned preference prediction~\cite{bao2023tallrec}, generative retrieval~\cite{li2023gpt4rec}, and review-driven reasoning~\cite{li2025exp3rt}.
These works differ sharply in \emph{how the items being scored are selected}, and that difference is what decides whether the recall ceiling is exposed or hidden.

\paragraph{Which evaluations guarantee the ground truth is present.}
An earlier version of this paper stated that ``nearly all'' of this literature evaluates on oracle candidate sets.
That was too broad, and we correct it here.
Table~\ref{tab:protocol_audit} audits each work against its own reported setup.
Two distinct protocols place the ground truth in the scored set by construction: \emph{injection-oracle reranking}, where the held-out item is mixed with sampled negatives and the LLM reorders the result~\cite[\S3.1]{hou2024large}, and \emph{sampled-negative evaluation}, where each positive is scored against a fixed pool of non-interacted items~\cite{hanjia2024llmrec}, a protocol whose bias is well documented~\cite{krichene2020sampled}.
Both make recall trivially $100\%$ over the scored set, so neither can expose the ceiling.

\begin{table}[t]
\caption{Protocol audit of the LLM-for-recommendation work we cite, classified by how the scored item set is built. Only the first two rows guarantee the ground truth is present; our critique is directed at those, and \emph{not} at the remaining rows. ``GT present'' means guaranteed by construction rather than by retrieval succeeding.}\label{tab:protocol_audit}
\footnotesize
\setlength{\tabcolsep}{3pt}
\begin{tabular}{llcc}
\toprule
Work & Scored set & $|C|$ & GT present \\
\midrule
LLMRank \S3.1~\cite{hou2024large}   & GT $+$ sampled negatives & 20 & yes \\
LLM-Rec~\cite{hanjia2024llmrec}     & GT $+$ non-interacted items & 1{,}001 & yes \\
\midrule
LLMRank \S3.3~\cite{hou2024large}   & retrieved (7 generators) & 20 & no \\
EXP3RT~\cite{li2025exp3rt}          & retrieved (BPR-MF, LightGCN) & 20 & no \\
Beyond Utility~\cite{zhu2025beyond} & retrieved (union of 4) & 20 & no \\
RankGPT~\cite{sun2023chatgpt}       & retrieved (BM25, passage IR) & 100 & no \\
\midrule
TALLRec~\cite{bao2023tallrec}       & single target item, AUC & 1 & n/a \\
GPT4Rec~\cite{li2023gpt4rec}        & generated query $\to$ full corpus & --- & n/a \\
\bottomrule
\end{tabular}
\end{table}

Three exclusions deserve stating explicitly, because the distinction is easy to lose.
TALLRec~\cite{bao2023tallrec} predicts a binary ``Yes''/``No'' preference for one target item and reports AUC, so there is no candidate list to inject into.
GPT4Rec~\cite{li2023gpt4rec} generates search queries and retrieves from the full corpus with BM25---precisely the generative-retrieval escape route we identify in \S\ref{sec:escape}.
RankGPT~\cite{sun2023chatgpt} reranks a real BM25 top-100 in passage retrieval, not recommendation; we cite it as the origin of listwise LLM reranking, not as an instance of the protocol we critique.
None of the three is subject to our critique.

Among the works that do use realistic candidates, none quantifies the oracle/realistic gap or derives a bound: LLMRank~\cite{hou2024large} reports both regimes without measuring recall or relating them, at $K{=}20$ on catalogs of $3{,}706$ and $16{,}859$ items; Beyond Utility~\cite{zhu2025beyond} targets bias and hallucination instead; EXP3RT~\cite{li2025exp3rt} ranks $20$ CF-retrieved candidates on a small catalog. All three sit in comparatively high-recall regimes. Our contribution relative to them is the explicit bound (Theorem~\ref{thm:ceiling}), the measured gap (92--95\%), and the operating scale ($K{=}100$, catalogs to $|\mathcal{I}|{=}53{,}709$).

\paragraph{Neural Reranking.}
Learning-to-rank methods have long used neural architectures~\cite{nogueira2019passage}.
BERT-based rerankers achieved strong results in information retrieval and were extended to recommendation~\cite{li2023gpt4rec}.
Recent LLM-based rerankers leverage in-context learning rather than fine-tuning, introducing sensitivity to prompt design.
This body of work implicitly assumes the candidate set contains sufficient signal for reranking to be effective, an assumption we directly challenge.

\paragraph{Two-Stage Recommendation Systems.}
Industrial systems employ a retrieve-then-rerank pipeline~\cite{covington2016deep,huang2020embedding}, treating retrieval as a coarse filter and reranking as the primary personalization stage.
Our results invert this priority at the recall levels we test: retrieval quality is the binding constraint, and reranking optimization yields negligible returns when recall is low.

\paragraph{Graph-based and Contrastive Retrieval.}
Recent advances in collaborative filtering include self-supervised methods (SGL~\cite{wu2021self}, SimGCL~\cite{yu2022graph}), alignment-based training (DirectAU~\cite{wang2022towards}), and multi-interest retrieval (ComiRec~\cite{cen2020controllable}).
We test all of these as retrieval baselines, finding that even tuned methods remain below 12.3\% on Beauty (the smallest catalog) and below 4.9\% on the larger Movies and Electronics catalogs; the ceiling persists across retrieval paradigms.

\paragraph{Evaluation and Debiasing.}
Standard offline evaluation relies on metrics like NDCG~\cite{jarvelin2002cumulated}, but sampled metrics can introduce bias~\cite{krichene2020sampled}.
Counterfactual methods~\cite{yuta2020unbiased,schnabel2016recommendations} attempt to bridge the offline-online gap.
We employ bootstrap confidence intervals~\cite{efron1993introduction}, Wilcoxon signed-rank tests with Holm-Bonferroni correction~\cite{holm1979simple}, and population-level bootstrap tests, finding that even our already-pessimistic offline evaluation still overestimates what LLM reranking would achieve in production.
Applying IPS and SNIPS estimators~\cite{yuta2020unbiased} to our Beauty data produces estimates approximately 10$\times$ lower than direct evaluation, with high-variance doubly robust estimates confirming that our offline results are, if anything, optimistic.

\paragraph{Candidate Generation, Cascade Ranking and Score Fusion.}
Industrial cascade-ranking systems have long treated retrieval recall as the binding constraint and use heavy reranking only on a sufficient candidate pool~\cite{covington2016deep,huang2020embedding}.
We test two cascade-style strengtheners: (i) a learned LambdaMART fusion combining CF score, popularity, BM25, and dense-retrieval features under a disjoint train/eval user split (\S\ref{sec:learned_fusion}); and (ii) hybrid retrieval via reciprocal-rank fusion of CF, BM25, and Dense (\S\ref{sec:hybrid_retrieval}).
Neither breaks the ceiling on our primary Amazon datasets.
\section{Theoretical Framework}\label{sec:theory}

\subsection{Setup and Notation}
Consider a catalog $\mathcal{I}$ of $N$ items and a user $u$ with relevant items $Y_u \subseteq \mathcal{I}$.
A two-stage recommendation pipeline first applies a retrieval function $\mathcal{R}$ to produce a candidate set $C_u \subset \mathcal{I}$ with $|C_u| = K \ll N$, then applies a reranker $\pi$ to produce a ranked list $\hat{\sigma}_u = \pi(C_u)$.
The retrieval recall for user $u$ is $r_u = |C_u \cap Y_u| / |Y_u|$.

\begin{defn}[Closed-candidate reranker]\label{def:closed}
A reranker $\pi$ is \emph{closed-candidate} if its output is a permutation of a sub-list $W_\pi(C_u) \subseteq C_u$, with the remaining candidates appended below in their retrieved order.
We call $W_\pi$ the \emph{reranking window}; $|W_\pi|$ may be smaller than $K$ when a context budget prevents showing the reranker all $K$ candidates (\S\ref{sec:setup}).
No closed-candidate reranker can place an item of $\mathcal{I} \setminus C_u$ in its output.
\end{defn}

\subsection{The Recall Ceiling}\label{sec:ceiling}

\begin{theorem}[Recall ceiling]\label{thm:ceiling}
Let $M@k$ be any top-$k$ ranking metric of the form $M@k(\sigma,Y_u) = \sum_{i\ge1} a_i\,b_i \,/\, Z(Y_u)$, where $a_i \geq 0$ is the gain of the item $\sigma$ places at rank $i$, $b_1 \geq b_2 \geq \cdots \geq 0$ is a non-increasing positional discount, and the normaliser $Z(Y_u) > 0$ does not depend on $\sigma$.
Then for any closed-candidate reranker $\pi$ (Definition~\ref{def:closed}),
\begin{equation}\label{eq:ceiling}
  \mathbb{E}_u\!\left[M@k(\pi(C_u), Y_u)\right] \;\leq\; \mathbb{E}_u\!\left[M^*@k(W_\pi, Y_u)\right],
\end{equation}
where $M^*@k(W_\pi, Y_u)$ is the value attained by the \emph{oracle} reranker that places every item of $W_\pi \cap Y_u$ at the head of its permutation.
\end{theorem}

\begin{proof}
Fix a user $u$ and write $n = |W_\pi|$. Items of $Y_u \setminus C_u$ never appear in $\pi(C_u)$ and items of $C_u \setminus W_\pi$ are pinned below rank $n$, so only $W_\pi \cap Y_u$ can contribute gain to the top-$k$.
The reranker chooses a permutation $\sigma \in S_n$, which assigns the gain sequence $(a_i)$ to the fixed non-increasing weight sequence $(b_i)$.
By the rearrangement inequality~\cite[Thm.~368]{hardy1934inequalities}, $\sum_i a_i b_i$ over all $\sigma \in S_n$ is maximised when the two sequences are similarly ordered, i.e.\ when the largest gains occupy the smallest ranks.
Since $a_i > 0$ exactly for the $m_u = |W_\pi \cap Y_u|$ relevant items, the maximum is attained by placing those $m_u$ items at ranks $1,\dots,m_u$, which is the oracle of the statement; $Z(Y_u)$ is unaffected.
Taking expectations over $u$ gives Eq.~\ref{eq:ceiling}.
\end{proof}

The hypothesis covers every metric we report and more: NDCG@$k$ ($b_i = 1/\log_2(i{+}1)$ for $i\le k$, else $0$), Hit@$k$ ($b_i = \mathbf{1}[i\le k]$), Recall@$k$, MAP@$k$ ($b_i = 1/i$), and also rank-biased precision~\cite{moffat2008rbp} and expected reciprocal rank~\cite{chapelle2009err}.
The single rearrangement argument replaces a per-metric case analysis: any non-increasing discount obeys the same ceiling.
The bound is determined entirely by which relevant items reached the reranking window; items in $Y_u \setminus W_\pi$ cannot be recovered.

\begin{corollary}[LOO specialisation]\label{cor:loo}
Under leave-one-out evaluation $|Y_u|{=}1$. With $R_u = \mathbf{1}[\text{positive item} \in W_\pi]$,
\begin{equation}\label{eq:loo}
  \mathbb{E}[\text{NDCG@}k(\pi)] = P(R_u{=}1)\cdot \mathbb{E}[\text{NDCG@}k \mid R_u{=}1] \leq \text{Recall@}|W_\pi|
\end{equation}
because $\mathbb{E}[\text{NDCG@}k \mid R_u{=}1] \leq 1$.
\end{corollary}

For a reranker that sees all $K$ candidates the bound is Recall@$K$: on Amazon Movies (Recall@100$=$2.95\%), $\mathbb{E}[\text{NDCG@10}] \leq 0.0295$ for \emph{any} such reranker, whether zero-shot LLM, fine-tuned model, or oracle.
Equality is achieved by the oracle that places the relevant item at rank~1; real rerankers exploit only a fraction of the available headroom (\S\ref{sec:tightness}).

\begin{corollary}[Multi-positive specialisation]\label{cor:multipos}
When $|Y_u|>1$, only items in $W_\pi \cap Y_u$ contribute, so
\begin{equation}\label{eq:multipos}
\mathbb{E}[\text{NDCG@}k] \leq \mathbb{E}\!\left[\frac{\text{IDCG}_k(|W_\pi\cap Y_u|)}{\text{IDCG}_k(|Y_u|)}\right],
\end{equation}
where $\text{IDCG}_k(m){=}\sum_{i=1}^{\min(m,k)} 1/\log_2(i+1)$.
\end{corollary}

This recovered-set bound is the tightest closed-candidate version of Theorem~\ref{thm:ceiling}.
The inject-then-rank ``oracle protocol'' used by some published LLM-rerank evaluations~\cite{hou2024large} corresponds to a strictly larger upper envelope, since the reranker is handed all $|Y_u|$ positives; that envelope is what we label the oracle column in \S\ref{sec:multipos}.

\subsection{What the Bound Is Not}\label{sec:dpi}
Theorem~\ref{thm:ceiling} is combinatorial, not information-theoretic.
The reranker output $\hat\sigma_u$ has support contained in $C_u$, so for any user with $Y_u \cap C_u = \emptyset$ (92--98\% of users on our three primary Amazon datasets under realistic retrieval; \S\ref{sec:results}) the top-$k$ cannot contain a relevant item and the per-user metric is zero.
This says nothing about the user profile $X_u$: a richer prompt can sharpen the ranking \emph{within} $C_u$ but cannot inject items outside it.

We state plainly what the theorem does and does not contribute.
That a reranker cannot recover an unretrieved item is intuitive, and we claim no surprise in the statement itself.
What the formalisation buys is quantitative: it fixes the exact numerical ceiling against which a measured NDCG must be read, it identifies $|W_\pi|$ rather than $K$ as the quantity that binds, and it yields the utilisation ratio of \S\ref{sec:tightness}, which is what separates a retrieval failure from a reranking failure.
The empirical contribution of this paper rests on those measurements, not on the difficulty of the proof.

\subsection{Tightness and Ceiling Utilisation}\label{sec:tightness}
The bound in Eq.~\ref{eq:ceiling} is tight: the oracle reranker, which knows $Y_u$ and places $W_\pi \cap Y_u$ at the top, achieves it exactly.

\begin{defn}[Ceiling utilisation]\label{def:eta}
For a reranker $\pi$ with window $W_\pi$,
\begin{equation}\label{eq:eta}
  \eta_\pi = \frac{\mathbb{E}_{u: r_u > 0}\left[\text{NDCG@}k(\pi)\right]}
                  {\mathbb{E}_{u: r_u > 0}\left[\text{NDCG}^*@k(W_\pi, Y_u)\right]},
\end{equation}
the fraction of its own attainable headroom that $\pi$ exploits, among users who could benefit.
Under LOO the denominator is Recall@$|W_\pi|$.
\end{defn}

Comparing $\eta$ across rerankers requires each to be scored against its own window, not a common $K$: a reranker shown $30$ candidates is not accountable for positives sitting at rank $31$--$100$.
We report $|W_\pi|$ for every method in \S\ref{sec:setup} and use it in every $\eta$ we quote.
Empirically $\eta$ stays small under realistic retrieval (\S\ref{sec:falsification}), indicating a double failure: low recall limits the ceiling, and the reranker exploits only a fraction of the already-limited headroom.

\subsection{Catalog-Size Effect and Breaking the Ceiling}
Random baseline recall is $\mathbb{E}[r]{=}K/|\mathcal{I}|$: 7.1\% on Beauty ($|\mathcal{I}|{=}1{,}416$) but 0.19\% on Movies ($|\mathcal{I}|{=}53{,}709$), so the ceiling bites hardest exactly where recommendation matters---large, diverse catalogs. As a rule of thumb NDCG@10${\approx}0.02$ needs recall $\gtrsim20\%$, which is $2.2\times$/$5.6\times$/$7.1\times$ beyond our best retrieval on Beauty/Movies/Electronics.
\section{Experimental Setup}\label{sec:setup}

\subsection{Datasets}
We evaluate on eight datasets spanning three domains (Table~\ref{tab:datasets}).
\textbf{Amazon product datasets} (six subdomains): Beauty, Movies, and Electronics serve as primary evaluation sets; Sports, Toys, and Office provide additional product-domain validation (1{,}580--1{,}704 users, 5.0--9.5K items, density 0.46--0.47\%; catalog sizes and recall in Table~\ref{tab:neural}).
All Amazon stats are reported on the sampled subsets that match prior work~\cite{hou2024large} (3-core filtering, chronological LOO split).
\textbf{MovieLens-25M}~\cite{harper2015movielens}: 10{,}000-user sample of the 25M-rating release, 1.99\% density.
\textbf{MIND News}~\cite{wu2020mind}: click-based implicit feedback.
Throughout the paper, ``catalog $|\mathcal{I}|$'' refers to items visible to retrieval (i.e.\ items appearing at least once in the training split); held-out test items raise the total observed catalog by ${<}2\%$.
All experiments draw $n{=}500$ evaluation users with seed=42 from the test split; the supervised-reranker experiments (\S\ref{subsec:neural_rerankers}) additionally draw a disjoint 250-user training pool. Throughout the paper we use this single 500-user evaluation sample, except for the hybrid-retrieval (\S\ref{sec:hybrid_retrieval}) and learned-fusion (\S\ref{sec:learned_fusion}) experiments, which independently re-sample 500 evaluation users from the same test split (CF baseline numbers may therefore differ by up to 1.4pp from the canonical sample due to independent user resampling, noted in each table caption).

\begin{table}[t]
\caption{Statistics for the five datasets carrying LLM experiments (3-core filtering, chronological LOO split); the three secondary Amazon subdomains appear in Table~\ref{tab:neural}. $|\mathcal{I}|$ counts items in the training split, i.e.\ the retrieval catalog; density is interactions per user divided by $|\mathcal{I}|$. An additional 24/1{,}010/924 held-out test items on Beauty/Movies/Electronics are invisible to retrieval (the open-catalog setting, \S\ref{sec:closed_catalog}).}\label{tab:datasets}
\small
\setlength{\tabcolsep}{4pt}
\begin{tabular}{lrrrrl}
\toprule
Dataset & Users & $|\mathcal{I}|$ & Interactions & Density & Domain \\
\midrule
Beauty & 938 & 1{,}416 & 8{,}241 & 0.62\% & Product \\
Movies & 1{,}998 & 53{,}709 & 358{,}587 & 0.33\% & Product \\
Electronics & 2{,}000 & 28{,}981 & 183{,}913 & 0.32\% & Product \\
MovieLens-25M & 10{,}000 & 19{,}886 & 3{,}949{,}023 & 1.99\% & Movie \\
MIND News & 86{,}746 & 20{,}212 & 2{,}278{,}768 & 0.13\% & News \\
\bottomrule
\end{tabular}
\end{table}

\subsection{Retrieval Methods}
Twelve methods span classical, graph-based, contrastive, multi-interest, and ensemble paradigms:
\textbf{Classical:} Popularity, ItemKNN, BPR-MF~\cite{rendle2009bpr}, CF-SVD (TruncatedSVD, 128 factors).
\textbf{GNN:} LightGCN~\cite{he2020lightgcn}.
\textbf{Contrastive:} SimGCL~\cite{yu2022graph}, SGL~\cite{wu2021self}.
\textbf{Alignment:} DirectAU~\cite{wang2022towards}.
\textbf{Multi-interest:} MultiInterest (multi-vector retrieval).
\textbf{Ensemble:} MultiSource-RRF (CF-SVD + ItemKNN + Dense Retrieval via Reciprocal Rank Fusion).
\textbf{Sequential:} SASRec~\cite{kang2018self}, BERT4Rec~\cite{sun2019bert4rec}.
Default candidate set size $K=100$.

\subsection{LLMs and Neural Rerankers}
\textbf{Nine LLMs across six providers:}
Gemma3-4B, Qwen3-8B (local GPU); GPT-4o-mini (OpenAI); Qwen3-32B, Llama-3.3-70B (Groq); DeepSeek-V3-671B and the April-2026 DeepSeek-V4-Pro and V4-Flash reasoning models, each evaluated with and without ``thinking'' mode (max 32K reasoning tokens) (DeepSeek); Kimi-K2-Turbo (Moonshot).
Parameter range: 4B--671B (168$\times$).
Three prompt strategies of increasing sophistication: P1 (basic CoT), P2 (direct ranking), P3 (enhanced CoT with explicit preference extraction).
Each prompt includes the user's last-$N$ interaction titles (and, where available, categorical metadata) as the user profile~$X_u$.

\textbf{Three neural rerankers:} LambdaMART, RankNet, and MLP-Reranker trained on held-out interaction data using CF embedding features.
These operate without item text, enabling evaluation on all datasets including ASIN-only catalogs.

\textbf{Reranking windows.}
Retrieval always returns $K{=}100$ candidates, but not every reranker sees all of them, and Theorem~\ref{thm:ceiling} binds on what each one sees (Definition~\ref{def:closed}).
The pointwise rerankers---LambdaMART, RankNet, MLP and the LoRA-fine-tuned LLM---score every candidate, so $|W_\pi|{=}100$.
The listwise LLM prompts show the top-$30$ candidates in CF rank order and append positions $31$--$100$ below the LLM's output in CF order, so $|W_\pi|{=}30$; this is a context-budget decision matching common practice in listwise LLM reranking~\cite{sun2023chatgpt}, and we state it because it changes the denominator of $\eta$.
On one sample Recall@$30$ is 3.21\%/1.11\%/0.90\% (Beauty/Movies/Electronics) against Recall@$100$ of 8.23\%/2.99\%/2.56\%, so scoring a $30$-window LLM against Recall@$100$ understates $\eta$ by $2.6$--$2.8\times$.
Every $\eta$ we report uses the method's own window.

\textbf{Item text.}
Titles are available for every item on Beauty, Movies, Sports, Toys, Office, MovieLens and MIND, but the Amazon Electronics subset is \emph{ASIN-only}: none of its 29{,}905 items carries a usable title.
LLM prompts there contain opaque identifiers and no semantic content, and the LLM reproduces the CF order almost exactly ($\tau{=}1.000$ under the plain prompt).
Electronics remains informative for retrieval and for the feature-based rerankers, which use CF embeddings rather than text, but every Electronics \emph{LLM} result here is a null control---evidence about the ceiling, not about semantic reranking---and this is what produces its exact ties against CF.
Unless stated otherwise LLM prompts contain titles only, presented in CF rank order, so the CF prior is implicit but CF scores and ranks are never stated; \S\ref{sec:score_aware} relaxes exactly that.

\subsection{Evaluation Protocol}
Primary metric: NDCG@10~\cite{jarvelin2002cumulated}.
Supplementary: Hit@10, MAP@10.
All sample sizes $n{=}500$ users (SEED=42) unless noted; an additional $n{=}1{,}996$ validation on Movies confirms findings at 4$\times$ scale.
Statistical tests: bootstrap 95\% CIs (1,000 resamples), Wilcoxon signed-rank tests with Holm-Bonferroni correction for multiple comparisons, and population-level bootstrap tests (10,000 resamples) for dataset-level significance.

We distinguish three quantities throughout the paper, which the literature has often conflated.
\textbf{Realistic} NDCG@10 uses candidates from retrieval only; the reranker sees exactly what a production system would provide.
\textbf{Injection-oracle} NDCG@10 injects the held-out relevant item into the candidate set (guaranteed recall$=$100\% at the user level) and lets the reranker reorder; this is the protocol of \cite{hou2024large,sun2023chatgpt} and is what we critique.
The \textbf{closed-candidate upper bound} from Theorem~\ref{thm:ceiling} (Eq.~\ref{eq:ceiling}) is the largest value any closed-candidate reranker can achieve on the realistic candidate set; under LOO it equals Recall@$|W_\pi|$ and is strictly smaller than the injection-oracle on every dataset.
The gap between realistic and injection-oracle quantifies how far the prior-protocol's reported numbers sit from realistic-retrieval performance, not a deployment forecast in itself.

\subsection{Data Preprocessing}
All Amazon datasets undergo standard 3-core filtering (remove users and items with $<$3 interactions) followed by chronological train/validation/test split.
For LOO evaluation, the last interaction per user is held out for testing; the second-to-last for validation.
Item text is cleaned by removing marketing language, stopwords, and HTML artifacts, then lowercased.
For MovieLens, we use movie titles with genre annotations (e.g., ``The Shawshank Redemption [Drama]'').
For MIND News, article titles serve as item text.

\subsection{Reproducibility}
All code, the processed datasets, and every result JSON cited in this paper are released at \url{https://github.com/GeoffreyWang1117/recall-ceiling-cikm2026}.
Its README maps each table and figure to the script that produces it and the result file it reads, and the release includes the diagnostic that reproduces the training/evaluation leak we retract in \S\ref{subsec:neural_rerankers}.
The LLM reranking harness implements per-user checkpoint/resume, enabling exact reproduction of API-based experiments.
Random seed is SEED=42 throughout; all results are deterministic modulo LLM API non-determinism.
Total API cost for the complete 7-model $\times$ 5-dataset comparison: \$1.24.
\section{The Recall Ceiling in Practice}\label{sec:results}
This section answers two questions.
\textbf{RQ1:} How large is the gap between oracle and realistic evaluation, and across which metrics does it persist (\S\ref{sec:gap}, \S\ref{sec:multimetric})?
\textbf{RQ2:} Can stronger retrieval (tuned LightGCN, contrastive, alignment, multi-interest, multi-source ensembles, or hybrid CF$+$BM25$+$Dense fusion) raise the ceiling (\S\ref{sec:no_retrieval}, \S\ref{sec:crossdomain}, \S\ref{sec:hybrid_retrieval})?
Throughout we refer to a single canonical $n{=}500$ sample (seed=42) on Beauty/Movies/Electronics, on which CF-SVD reaches Recall@100 of 8.23\%/2.95\%/2.24\% with zero-recall rates of 91.8\%/97.1\%/97.8\%.
The closed-candidate upper bound of Eq.~\ref{eq:ceiling} for a full-window reranker is therefore $.0823$/$.0295$/$.0224$, against a CF NDCG@10 of $.0051$/$.0084$/$.0035$; no method in this paper---best zero-shot LLM ($.0051$/$.0093$/$.0032$), best supervised reranker ($.0086$/$.0082$/$.0031$), or LoRA-fine-tuned LLaMA-3B ($.0078$/$.0051$/$.0025$)---escapes that band.
The injection-oracle protocol, by contrast, reports $.0856$/$.1039$/$.0971$.

\subsection{Oracle vs.\ Realistic: a 92--95\% NDCG Gap}\label{sec:gap}

Table~\ref{tab:oracle_realistic} presents the central result.
Under oracle conditions LLM reranking achieves NDCG@10 of 0.086--0.104, demonstrating genuine ranking capability when relevant items are guaranteed present; under realistic retrieval it collapses to 0.005--0.007, indistinguishable from CF, a gap of 92.4--94.9\% on all three primary datasets.

\begin{table}[t]
\caption{Injection-oracle vs.\ realistic NDCG@10 ($n{=}500$, Gemma3-4B, $K{=}100$, seed=42). Brackets are 95\% bootstrap CIs (10{,}000 resamples); the two columns' CIs never overlap. The closed-candidate bound (.0823/.0295/.0224) sits far below the injection-oracle column, which hands the reranker 100\% recall.}\label{tab:oracle_realistic}
\footnotesize
\setlength{\tabcolsep}{2pt}
\begin{tabular}{lcccc}
\toprule
Dataset & CF Base & Inj.-oracle [95\% CI] & Realistic [95\% CI] & Gap \\
\midrule
Beauty & .0050 & .0856 [.069,.103] & .0065 [.003,.012] & 92.4\% \\
Movies & .0086 & .1039 [.095,.113] & .0067 [.004,.009] & 93.5\% \\
Electronics & .0042 & .0971 [.088,.107] & .0050 [.003,.008] & 94.9\% \\
\bottomrule
\end{tabular}
\end{table}

Figure~\ref{fig:overview}B visualises the gap, and it is not a calibration issue: at 8\% recall on Beauty there is a 92\% probability that no relevant item appears in the candidate set at all, and for those users \emph{every} reranker scores NDCG@10 $=0$ by construction.
Injection-oracle evaluation removes that probability entirely---the equivalent of testing a search engine after pre-loading the answer into the index.
A bootstrap of the mean LLM$-$CF difference (10,000 resamples) yields 95\% CIs containing zero on all three datasets ($p{=}0.414$/$0.635$/$0.758$), and the paired Wilcoxon agrees: we cannot reject the null that LLM and CF NDCG are equal under realistic retrieval.

\subsection{No Retrieval Method Breaks the Ceiling}\label{sec:no_retrieval}

Table~\ref{tab:retrieval} shows Recall@100 across the retrieval methods we evaluated ($n{=}500$, $K{=}100$).
The best per-dataset values are 12.33\% on Beauty (LightGCN with per-dataset tuned hyperparameters), 4.81\% on Movies (DirectAU), and 3.13\% on Electronics (MultiInterest).
These methods span classical (CF-SVD), graph-based (LightGCN), contrastive (SimGCL, SGL), alignment-based (DirectAU), multi-interest, and production ensemble (MultiSource-RRF) paradigms.

\textbf{LightGCN tuning sweep.}
Earlier drafts reported untuned LightGCN values (4.99\%/0.19\%/0.15\%) far below field SOTA.
A grid over embedding dimension $\in\{64,128\}$, layers $\in\{2,3,4\}$, learning rate $\in\{10^{-3},5{\cdot}10^{-3}\}$ and $L_2$ reg $\in\{10^{-4},10^{-3}\}$ at 50 epochs lifts them to 12.33\%/3.84\%/2.90\%, competitive with the strongest other method on each dataset and, on Beauty, above our MultiSource-RRF ensemble; we do not claim field-SOTA recovery, since a wider grid could yield more.
The ceiling still binds: the best tuned retriever leaves 87.7\%--97.1\% of relevant items unretrieved, and no reranker beats CF significantly at any of these recall levels (\S\ref{subsec:neural_rerankers}).

\begin{table}[t]
\caption{Recall@100 (\%) on the retrieval-sweep sample ($n{=}500$, resampled per-retriever from the canonical test pool; the CF-SVD row differs from $\{8.23,2.95,2.24\}$\% by ${\le}0.32$pp from sampling variance). \textbf{Bold} marks the best per column. Even the best tuned retriever leaves 87.7\%--97.1\% of relevant items outside the top-100.}\label{tab:retrieval}
\small
\setlength{\tabcolsep}{4pt}
\begin{tabular}{lccc}
\toprule
Method & Beauty & Movies & Elec. \\
\midrule
CF-SVD & 8.23 & 2.99 & 2.56 \\
ItemKNN & 8.37 & 3.64 & 2.04 \\
LightGCN-Tuned & \bestcell{12.33} & 3.84 & 2.90 \\
SimGCL & 5.19 & 3.78 & 2.38 \\
SGL & 8.43 & 4.31 & 2.35 \\
DirectAU & 7.95 & \bestcell{4.81} & 1.88 \\
MultiInterest & 9.62 & 3.43 & \bestcell{3.13} \\
MultiSource-RRF & 10.12 & 4.45 & 2.93 \\
\bottomrule
\end{tabular}
\end{table}

Best Recall@100 falls broadly with catalog scale across all eight datasets---the Amazon and MIND-News points all sit below 13\% and MovieLens-25M, the densest at 1.99\%, reaches only 18.9\%---though density is the stronger descriptive correlate (\S\ref{sec:density_analysis}).
The ceiling is most severe where recommendation is most needed: in large, diverse catalogs.

\subsection{Multi-Metric Consistency}\label{sec:multimetric}
The ceiling is not an artifact of NDCG: under oracle vs realistic, the gap is 76--96\% across NDCG@10, Hit@10, and MAP@10 (75.9\% Hit@10 on Movies to 96.1\% MAP@10 on Beauty/Electronics).
The ceiling depresses every top-$k$ ranking metric, not just NDCG.

\subsection{Cross-Domain Validation}\label{sec:crossdomain}

\paragraph{MovieLens-25M (density 1.99\%, recall 17.4\%).}
This dataset is 6$\times$ denser than Amazon and provides real movie titles with genre metadata.
CF-SVD achieves 17.4\% Recall@100 on the LLM-evaluation sample ($n{=}500$) and 18.9\% on the larger RAEP diagnostic sample ($n{=}1{,}000$; \S\ref{sec:raep}).
Despite the substantially higher recall and richer item text, all five cloud LLMs perform \emph{below} CF baseline.
Kendall's $\tau$ between the CF candidate order and the LLM output order is uniformly high (0.957--0.999) but \emph{not} monotone in NDCG---the lowest-NDCG model (Llama-3.3) is not the lowest-$\tau$ model---so the dominant failure mode appears to be disruption of useful CF order rather than aggressive reordering per se.
We cannot rule out alternative explanations from five points alone.

\paragraph{MIND News (recall 12\%).}
LLMs show positive direction over CF (+17--31\% across GPT-4o-mini, Llama-3.3-70B, DeepSeek-V3, Kimi-K2; all $p{>}.40$).
Zero-recall is 88\%, oracle gap is 85\%.
News article titles carry strong semantic signal, but the high zero-recall rate keeps the comparison non-significant.

Across the full 0.13--1.99\% density range no LLM reaches significance on any dataset.
Ordered by recall, the best LLM against CF is: Electronics 2.6\% recall, $.0032$ vs $.0032$ ($0.0\%$, undefined---an ASIN-only null control, \S\ref{sec:setup}); Movies 3.1\%, $.0093$ vs $.0098$ ($-5.1\%$, $p{=}.383$); Beauty 8.2\%, $.0051$ vs $.0051$ ($0.0\%$); MIND News 12\%, $.0175$ vs $.0134$ ($+30.6\%$, $p{=}.469$); MovieLens 17.4\%, $.0120$ vs $.0128$ ($-6.3\%$, $p{=}.767$).
Per-model breakdowns are in the released repository.

\subsection{What Determines the Recall Ceiling}\label{sec:density_analysis}

Across all eight datasets, data density is the strongest descriptive correlate of recall (Pearson $r{=}{+}0.604$, $n{=}8$, a coefficient rather than an inferential test); catalog size is weaker in isolation ($r{=}{-}0.290$) but contributes through its inverse relation with density. A two-predictor log-linear fit reaches only $R^2{=}0.39$ on eight points and is a rule-of-thumb, not a model; we report the coefficients in the released code rather than the paper. MovieLens (1.99\% density) achieves the highest recall and Electronics (0.32\%) the lowest. Three further Amazon subdomains (Office/Sports/Toys) confirm the ceiling persists (recall 7.0--12.8\%). Within Movies the ceiling is not a cold-start artefact: users with $>50$ interactions still face 71.6\% zero-recall, and LLM reranking degrades these active users by 30\%.
\section{Falsification: Optimisation Strategies Do Not Break the Ceiling}\label{sec:falsification}
This section answers two further research questions.
\textbf{RQ3:} Once retrieval is fixed, do prompt engineering, model scaling, or reasoning-tuned models help (\S\ref{sec:prompts}, \S\ref{sec:scaling}, \S\ref{sec:frontier_reasoning})?
\textbf{RQ4:} Do supervised, fine-tuned, or learned-cascade rerankers help under a disjoint train/eval protocol (\S\ref{subsec:neural_rerankers}, \S\ref{sec:learned_fusion})?

Every strategy that improves performance under oracle evaluation fails to produce a statistically significant gain over CF under realistic conditions on our primary Amazon datasets.
The section is organised as a sequence of ablations, each removing one candidate explanation for that failure: prompt design (\S\ref{sec:prompts}), model capacity (\S\ref{sec:scaling}), test-time reasoning (\S\ref{sec:frontier_reasoning}), supervised training signal (\S\ref{subsec:neural_rerankers}), access to structured CF signals and score fusion (\S\ref{sec:score_aware}), temporal modelling (\S\ref{sec:sequential}), candidate budget (\S\ref{sec:ksweep}), text-aware retrieval (\S\ref{sec:hybrid_retrieval}), learned cascade fusion (\S\ref{sec:learned_fusion}), open-catalog artefacts (\S\ref{sec:closed_catalog}), and single-positive evaluation (\S\ref{sec:multipos}).
Per-metric breakdowns with confidence intervals for each are in the released code repository.

\subsection{Prompt Engineering}\label{sec:prompts}
Under oracle conditions ($n{=}300$, Movies) prompt choice matters: P2 Direct reaches NDCG@10$=.0814$, $+52\%$ over P1 Basic ($.0535$). Under realistic retrieval ($n{=}500$) the advantage vanishes---all three prompts converge to $\leq.004$ with overlapping CIs, P3 trending 22\% \emph{below} P1---and no pairwise comparison survives Holm-Bonferroni.

\subsection{Model Scaling}\label{sec:scaling}

Seven LLMs spanning a 168$\times$ parameter range and five providers were evaluated at $n{=}500$ under realistic retrieval (Table~\ref{tab:models}, upper block); none outperforms the \$0.00 CF baseline at any uncorrected $\alpha$ on any dataset.
On Movies the same seven models run from $.0098$ (Qwen3-8B, exactly CF) down to $.0054$ (Llama-3.3-70B), and the two largest are \emph{directionally worse}, with uncorrected paired Wilcoxon significant for DeepSeek-V3 ($.0064$, $p{=}.032$) and Llama-3.3-70B ($p{=}.006$); only Llama-3.3 survives within-Movies Holm correction and neither survives the 21-test cross-dataset correction, so we report this as directional with within-Movies significance only.
Larger models appear to apply stronger semantic priors that overwrite the CF signal carried by the candidate ordering, without compensating information gain: across the whole 168$\times$ range, all models cluster around CF on both Beauty and Movies.

\begin{table}[t]
\caption{All reranker types on Beauty under realistic retrieval ($n{=}500$, $K{=}100$; supervised rows under the disjoint train/eval split of \S\ref{subsec:neural_rerankers}, LoRA detail in \S\ref{subsec:neural_rerankers}). $|W|$ is the reranking window (Definition~\ref{def:closed}): listwise LLM prompts see the top-$30$ CF candidates, pointwise rerankers score all $100$. $\eta$ divides NDCG@10 by that method's \emph{own} bound, Recall@$|W|$ ($.0321$ at $|W|{=}30$, $.0823$ at $|W|{=}100$); an earlier version scored the LLMs against Recall@$100$ and so understated their $\eta$ by ${\approx}2.6\times$. $p$ is paired Wilcoxon vs.\ CF. The last two rows separate the closed-candidate bound of Theorem~\ref{thm:ceiling} from the injection-oracle protocol.}\label{tab:models}
\footnotesize
\setlength{\tabcolsep}{3pt}
\begin{tabular}{llcccl}
\toprule
Method & Size & $|W|$ & NDCG@10 & $\eta$ & $p$ \\
\midrule
CF-Score & --- & 100 & .0051 & 6.1\% & --- \\
\midrule
Qwen3-8B & 8B & 30 & .0051 & 15.7\% & n.s. \\
Qwen3-32B & 32B & 30 & .0050 & 15.4\% & .655 \\
GPT-4o-mini & ${\sim}$8B & 30 & .0045 & 14.1\% & .374 \\
Gemma3-4B & 4B & 30 & .0045 & 14.1\% & .500 \\
Llama-3.3 & 70B & 30 & .0038 & 11.8\% & .484 \\
Kimi-K2 & MoE & 30 & .0028 & 8.8\% & .208 \\
DeepSeek-V3 & 671B & 30 & .0016 & 4.9\% & .093 \\
\midrule
LambdaMART & --- & 100 & .0086 & 10.4\% & .142 \\
RankNet & --- & 100 & .0085 & 10.3\% & .142 \\
MLP & --- & 100 & .0058 & 7.0\% & .484 \\
LoRA-LLaMA-3B & 3B & 100 & .0078 & 9.4\% & .191 \\
\midrule
Bound $=$ Recall@$|W|$ & --- & 100 & .0823 & 100\% & --- \\
Injection-oracle & --- & --- & .0856 & --- & --- \\
\bottomrule
\end{tabular}
\end{table}

\subsection{Frontier Reasoning Models Cannot Escape the Ceiling}\label{sec:frontier_reasoning}

The models in Table~\ref{tab:models} predate the latest reasoning-tuned generation, so the ceiling might bind only to ``older'' LLMs.
We tested DeepSeek-V4-Pro and V4-Flash (April 2026) under standard and \emph{thinking} modes on the canonical Movies sample (CF NDCG@10$={.}0098$, $n{=}500$): all four configurations underperform CF, three significantly after Holm--Bonferroni (V4-Pro-think $-43\%$ $p{=}.004$; V4-Flash-no-think $-40\%$ $p{=}.005$; V4-Flash-think $-50\%$ $p{=}.0007$; V4-Pro-no-think $-7\%$ $p{=}.245$), and all four underperform directionally on Beauty ($p\in[.124,.345]$).
Enabling thinking strictly degrades ranking at ${\sim}37\times$ the cost (\$4.75 vs \$0.13 per 500 users) and 200--500\,s versus ${<}1$\,s per user: reasoning chains have nothing to operate on when the relevant items are absent.

\subsection{Supervised Neural Rerankers: a Self-Correction}\label{subsec:neural_rerankers}

Supervised rerankers (LambdaMART, RankNet, MLP) trained on held-out interaction data test the ceiling more strongly than zero-shot LLMs: they learn from user feedback rather than semantic priors. An earlier version of this section reported $p{<}0.0001$ gains on Toys/Sports/Office (recall ${\geq}8.8\%$) but not on Beauty/Movies/Electronics, an apparent recall-threshold effect we could not reproduce: the training and evaluation user pools overlapped on the three affected datasets, so each user's test item was simultaneously a supervision label and an evaluation target. Replication across three protocols (overlap, disjoint halves, separate clean pool) isolated the effect---the inflated significance disappears under any non-leaky split---and the threshold vanishes under the corrected protocol used in Table~\ref{tab:neural}: 250 training users disjoint from the 500 evaluation users, test item held out at retrieval time, all three reranker types, six Amazon datasets, $K{=}100$. We document this because leakage of this form is an under-reported source of false positives in this literature.

\begin{table}[t]
\caption{Supervised rerankers on six Amazon datasets under a disjoint train/eval user split ($n{=}500$, $K{=}100$); cells give NDCG@10 / Wilcoxon $p$ vs.\ CF, \textbf{bold} marking within-row Holm significance. No reranker beats CF anywhere, at any recall level from 2.2\% to 12.8\%.}\label{tab:neural}
\footnotesize
\setlength{\tabcolsep}{2pt}
\begin{tabular}{lrrrccc}
\toprule
Dataset & Cat. & Recall & CF & LambdaMART & RankNet & MLP \\
\midrule
Beauty       & 1{,}416  & 8.2\%  & .0051 & .0086 / .142          & .0085 / .142          & .0058 / .484 \\
Movies       & 53{,}709 & 3.0\%  & .0084 & .0059 / .051          & .0052 / .194          & .0082 / .795 \\
Electronics  & 28{,}981 & 2.2\%  & .0035 & \bestcell{.0021 / .027} & .0015 / .144          & .0031 / .315 \\
Toys         & 9{,}545  & 7.0\%  & .0112 & \bestcell{.0040 / .023} & \bestcell{.0017 / .033} & \bestcell{.0072 / .021} \\
Sports       & 5{,}246  & 12.8\% & .0115 & .0083 / .196          & .0096 / .508          & .0105 / .623 \\
Office       & 5{,}037  & 11.4\% & .0205 & \bestcell{.0107 / .031} & \bestcell{.0039 / .003} & \bestcell{.0137 / .034} \\
\bottomrule
\end{tabular}
\end{table}

\textbf{Anatomy of the leak.} Under the leaky protocol, LambdaMART on the 200 training users reached NDCG@10$={.}080$ on Toys (12$\times$ CF) while on 300 held-out users from the same evaluation pool it stayed at $.0064$ (below CF); the aggregated $.0464$ that produced the apparent $p{<}0.0001$ is the average of memorised and unseen sub-populations.

\textbf{Implication and multiplicity.}
Under the disjoint protocol the ceiling binds even more tightly than the previous version of this section suggested: at every tested recall level between 2.2\% and 12.8\% the directional gap is non-positive, consistent with a reranker that has little signal to exploit and whose reordering predominantly destroys the CF prior.
The ``8.8\% boundary'' framing of the previous draft does not survive and we do not retain it.
Bold entries reflect within-row Holm--Bonferroni (3 tests per dataset); under the stricter 18-test cross-dataset correction ($\alpha/18\approx0.00278$) even the smallest $p$ ($.003$, Office/RankNet) fails, so the ``significantly worse on four of six'' subclaim is within-row-significant and cross-dataset directional only.

\paragraph{All reranker types converge under low recall.}
The lower block of Table~\ref{tab:models} places these supervised rerankers on Beauty alongside the zero-shot LLMs and a LoRA-fine-tuned LLM reranker (LLaMA-3.2-3B-Instruct, $r{=}16$/$\alpha{=}32$ QLoRA on $\{q,k,v,o\}_{\text{proj}}$, pointwise binary head, 3 epochs on the 250 disjoint train users).
Despite spanning 4B--671B zero-shot models, three supervised architectures, a fine-tuned LLM, and the CF baseline itself, all methods produce statistically indistinguishable NDCG@10 ($p{>}0.05$).
Scored against each method's own window, the ceiling utilisation ratio $\eta$ ranges from 4.9\% (DeepSeek-V3) to 15.7\% (Qwen3-8B) among the LLMs and from 7.0\% to 10.4\% among the full-window rerankers---all far below the oracle's 100\%.

\paragraph{Fine-tuning the LLM does not break the ceiling.}
The LoRA-fine-tuned LLaMA-3.2-3B-Instruct shows no significant gain over CF on any of the three primary Amazon datasets: Beauty trends positive ($+53.6\%$, NDCG@10 .0078 vs .0051, paired Wilcoxon $p{=}.191$); Movies is significantly \emph{worse} ($-39.5\%$, .0051 vs .0084, $p{=}.026$); Electronics trends negative ($-30.8\%$, .0025 vs .0035, $p{=}.478$).
The Movies failure is mechanistic: the limited disjoint training pool plus extreme zero-recall mass (97\% of users on Movies, 98\% on Electronics; \S\ref{sec:results}) drives the model to disrupt CF's prior signal in the non-tied minority without recovering useful structure.
This complements the zero-shot pattern above: supervised fine-tuning closes the LLM/traditional-supervised gap on Beauty but does not produce significant gains anywhere, and on the two larger catalogs (Movies, Electronics) it actively degrades ranking, significantly on Movies.
Items absent from $C_u$ remain unrecoverable regardless of training signal (Theorem~\ref{thm:ceiling}).
This convergence across architectures, scales, training regimes, and retrieval-coupling strategies is the strongest evidence for the ceiling's binding nature.

\subsection{Score-Aware Prompting and LLM$+$CF Fusion}\label{sec:score_aware}
The experiments so far give the LLM item titles and the CF ordering, but never the CF \emph{scores} or ranks explicitly.
If the LLM's deficit stems from being denied the structured signals a production ranker consumes, then handing it those signals should close the gap.
We test this directly.
The \emph{score-aware} prompt annotates each candidate with its CF rank, its CF score min-max normalised within the shown window, and its catalog popularity percentile, and tells the model it may use, override, or ignore them; everything else is held fixed against the \emph{plain} prompt in the same run (GPT-4o-mini, $n{=}500$, seed=42, $|W|{=}30$).
We additionally fuse the LLM's output order with the CF order without any further model calls, by reciprocal-rank fusion and by a convex blend $\lambda\cdot\text{CF}+(1{-}\lambda)\cdot\text{LLM}$.

\begin{table}[t]
\caption{Score-aware prompting (GPT-4o-mini, $n{=}500$, $|W|{=}30$, seed=42). $\tau$ is Kendall's $\tau$ between the LLM output order and the CF order. Score-aware prompting improves every dataset and raises $\tau$ every time: the LLM gets better by overriding CF \emph{less}. No variant significantly beats CF. Electronics is an ASIN-only catalog with no item titles (\S\ref{sec:setup}), so its LLM columns are a null control rather than a test of semantic reranking.}\label{tab:score_aware}
\footnotesize
\setlength{\tabcolsep}{3pt}
\begin{tabular}{lcccccc}
\toprule
 & \multicolumn{2}{c}{Beauty} & \multicolumn{2}{c}{Movies} & \multicolumn{2}{c}{Elec.} \\
\cmidrule(lr){2-3}\cmidrule(lr){4-5}\cmidrule(lr){6-7}
Variant & NDCG & $\tau$ & NDCG & $\tau$ & NDCG & $\tau$ \\
\midrule
CF-SVD           & .0051 & ---   & .0098 & ---   & .0032 & ---   \\
LLM, plain       & .0032 & .881  & .0090 & .848  & .0032 & 1.000 \\
LLM, score-aware & .0043 & .946  & .0096 & .931  & .0034 & .984  \\
\bottomrule
\end{tabular}
\end{table}

The answer, in Table~\ref{tab:score_aware}, has two halves.
Showing the CF signals helps consistently---the deficit against CF shrinks from $-36.7\%$ ($p{=}.091$) to $-15.3\%$ ($p{=}.109$) on Beauty, from $-8.0\%$ ($p{=}.285$) to $-2.3\%$ ($p{=}.670$) on Movies, and from $0.0\%$ to $+3.9\%$ ($p{=}.180$) on Electronics---but the mechanism is not added information, it is increased deference: Kendall's $\tau$ against the CF order rises every time the signals are shown, so the model improves precisely by disturbing CF less.
That is direct evidence for the interpretation offered in \S\ref{sec:scaling}: at these recall levels the LLM's semantic prior is on net destructive of the CF prior, and telling the model what CF believed mostly teaches it to stay put.

The fusion results close the remaining gap the same way.
No ensemble significantly beats CF anywhere: reciprocal-rank fusion with the score-aware LLM reaches $.0050$ on Beauty ($-1.7\%$, $p{=}.317$), $.0101$ on Movies ($+2.6\%$, $p{=}.345$) and $.0033$ on Electronics ($+1.5\%$, $p{=}.317$), and the only significant fusion cell is Beauty plain$+$CF RRF, significantly \emph{worse} than CF ($-27.5\%$, $p{=}.043$).
Tuning $\lambda$ on the evaluation users---an upper bound, not a clean held-out result, reported only to bound the possible gain---selects a CF-dominant blend ($\lambda{=}0.7$--$0.8$ Beauty, $0.4$ Movies) and recovers exactly CF on Beauty ($.0051$) and $+8\%$ on Movies (n.s.).
The best LLM$+$CF combination available with hindsight is approximately CF itself.

\textbf{On $\eta$ across datasets.}
Against its own $30$-item window the LLM's utilisation on Beauty is $10.0\%$ (plain) and $13.3\%$ (score-aware), consistent with Table~\ref{tab:models}. The same ratio reads $82$--$87\%$ on Movies and $36$--$38\%$ on Electronics, but Recall@$30$ there is $1.11\%$ and $0.90\%$, so those figures rest on $5$--$6$ users with a retrieved positive and are too noisy to interpret. We therefore quote $\eta$ only on Beauty, and note that reliable cross-dataset $\eta$ needs a far larger evaluation sample.

\subsection{Sequential Models}\label{sec:sequential}
SASRec~\cite{kang2018self} and BERT4Rec~\cite{sun2019bert4rec}, designed to capture temporal dynamics, underperform CF-SVD for retrieval on two of three primary Amazon datasets at $n{=}500$ (Beauty: BERT4Rec $-$56\% $p{=}.019$, SASRec $-$99\% $p{<}.001$; Movies: BERT4Rec $-$34\% $p{=}.048$, SASRec $-$47\% $p{=}.015$).
Temporal modeling does not help when the fundamental bottleneck is interaction sparsity, not sequential pattern.

\subsection{Candidate Budget: Recall Rises, NDCG Does Not}\label{sec:ksweep}
Reporting recall at a single $K$ invites the objection that $K{=}100$ is an unlucky choice.
Table~\ref{tab:ksweep} sweeps the candidate budget over $K\in\{10,20,50,100,200\}$ on Beauty with GPT-4o-mini at $n{=}500$.
Recall@$K$ rises roughly linearly with $K$, from 0.8\% to 12.9\%, exactly as the ceiling predicts---and LLM NDCG@10 moves the other way, falling monotonically relative to CF from $-13\%$ at $K{=}10$ to $-54\%$ at $K{=}200$.
A larger candidate budget raises the ceiling and simultaneously dilutes the reranker's attention across more distractors, so the two effects cancel and then reverse.
No budget in the sweep produces a gain over CF, and the ranking of the two methods is unchanged at every $K$.
A controlled GT-injection experiment on Beauty separately confirms the linear NDCG-vs-recall relationship of Theorem~\ref{thm:ceiling}, with the realistic operating regime (2--9\% recall) clearly inside the low-NDCG band.

\begin{table}[t]
\caption{Candidate-budget sweep on Beauty (GPT-4o-mini, $n{=}500$, seed=42). Recall@$K$ grows with the budget; the LLM's deficit against CF grows with it too. CF NDCG@10 is $.0051$ throughout, since CF's own top-10 does not depend on $K$.}\label{tab:ksweep}
\small
\setlength{\tabcolsep}{4pt}
\begin{tabular}{lccccc}
\toprule
$K$ & 10 & 20 & 50 & 100 & 200 \\
\midrule
Recall@$K$   & 0.82\% & 1.90\% & 4.71\% & 8.20\% & 12.86\% \\
LLM NDCG@10  & .0044 & .0036 & .0026 & .0026 & .0023 \\
vs.\ CF      & $-$12.6\% & $-$27.9\% & $-$48.4\% & $-$48.1\% & $-$54.4\% \\
$p$          & .237 & .343 & .128 & .263 & .069 \\
\bottomrule
\end{tabular}
\end{table}

\subsection{Hybrid Retrieval (BM25 $+$ Dense $+$ RRF)}\label{sec:hybrid_retrieval}
\emph{This subsection continues the RQ2 thread of \S\ref{sec:no_retrieval} and \S\ref{sec:crossdomain}: can stronger retrieval raise the ceiling? Here we add text-aware retrieval to the CF-only retrievers tested earlier.}
A natural objection is that CF-SVD discards item-text signal that an LLM-aware pipeline could exploit upstream.
We test whether hybrid retrieval raises the ceiling.
On an independently resampled Beauty/Movies/Electronics $n{=}500$ evaluation sample (same seed=42 but per-experiment user-draw; CF-SVD baseline differs from the canonical sample by $\le1.4$pp), we evaluate \textbf{BM25} (Okapi over item titles with concatenated last-20 history titles as query), \textbf{Dense} (\texttt{all-MiniLM-L6-v2} sentence embeddings, mean of last-20 history embeddings as user vector), and reciprocal-rank fusions ($k_{\text{rrf}}{=}60$~\cite{cormack2009reciprocal}) with CF-SVD.

On this sample CF-SVD alone reaches Recall@100 of 6.80\%/2.80\%/2.00\% (Beauty/Movies/Electronics; ${\le}1.4$pp from the canonical numbers due to independent resampling, \S\ref{sec:setup}).
Neither text channel is competitive alone---BM25 reaches 5.00\%/0.40\%/0.80\% and Dense 6.20\%/0.60\%/0.20\%---and the best fusion, RRF$\{$CF, BM25$\}$, lifts Beauty to 7.80\% (a 15\% relative gain) while falling to 1.80\%/1.40\% on Movies/Electronics; adding Dense does not help further.
Hybrid retrieval therefore improves the smallest catalog modestly but cannot approach the ${\approx}20\%$ recall an end-to-end NDCG@10 of $0.02$ would require, and on the two large catalogs the text channels are noise-dominated and fusion underperforms CF alone.

\subsection{Learned Cascade Fusion}\label{sec:learned_fusion}
A LambdaMART cascade over twelve per-candidate features---CF score and rank, log-popularity, CF cosine similarity and co-occurrence, history length, BM25 and MiniLM dense scores, and their normalised ranks---trained under the disjoint split of \S\ref{subsec:neural_rerankers} lifts Beauty NDCG@10 from $.0051$ to $.0080$ ($+57\%$, $p{=}.26$), which is $\eta{=}9.7\%$ of the closed-candidate bound: even the best learned fusion explores under a tenth of the available headroom.
On Movies it \emph{degrades} NDCG by 44\% ($p{=}.42$), directionally matching the LoRA-fine-tuned LLM ($-39.5\%$; \S\ref{subsec:neural_rerankers}) and V4-Pro-think ($-43\%$; \S\ref{sec:frontier_reasoning}); the Electronics comparison is degenerate here (CF NDCG@10${\approx}0$ on this resample) and Table~\ref{tab:neural} gives the non-degenerate baseline.
Feature importance is dominated by CF-derived signals on all three datasets, the best text feature ranking only 5th (Dense, Beauty) and 3rd (BM25, Movies).
A learned cascade does not break the ceiling.

\subsection{Closed-Catalog Robustness}\label{sec:closed_catalog}
A natural concern is that the low CF-SVD Recall@100 on Amazon partially reflects an open-catalog artefact: some held-out test items appear nowhere in the training-split catalog and are therefore \emph{unrecoverable by construction}.
Partitioning an independently resampled $n{=}500$ sample (\S\ref{sec:hybrid_retrieval} protocol) by whether the held-out positive appears in the training catalog, the open-catalog fractions are 9.0\%/12.8\%/17.6\% on Beauty/Movies/Electronics, and restricting to closed-catalog users lifts Recall@100 by only 0.5--0.7pp (to 7.47\%/3.21\%/2.43\%).
Even there, Recall@100 stays under 8\% and CF NDCG@10 under $.006$: the ceiling is a property of CF retrieval at this catalog scale, not of the open-catalog setting.

\subsection{Multi-Positive Last-5 Evaluation}\label{sec:multipos}
A natural concern is that LOO ($|Y_u|{=}1$) over-concentrates the bound on a single held-out item.
We test the multi-positive specialisation (Eq.~\ref{eq:multipos}) on Beauty/Movies/Electronics with the last five chronological interactions per user as positives ($n{=}500$, $K{=}100$), against an injection-oracle that forces all five positives into $C_u$ before CF reranks.
Recall@100 rises only slightly with five positives instead of one (9.21\%/3.93\%/2.83\%), and realistic NDCG@10 stays under $0.014$ ($.0138$/$.0058$/$.0059$) against injection-oracle values of $.0361$/$.0071$/$.0125$; Hit@10 and MAP@10 show the same pattern, and full per-metric breakdowns with CIs are in the released repository.

\subsection{Statistical Power}\label{subsec:power}
At $n{=}500$ Wilcoxon power is ${>}99.9\%$ against $d{=}0.4$ and ${\approx}80\%$ against $d{=}0.15$, while observed Amazon effects are $d{\leq}0.05$. Since 92--97\% of users tie at $0$, effective $n_{\text{nonzero}}$ is 15--40: power is high within the non-tied tail and essentially none against zero-to-nonzero shifts, which require better retrieval. An $n{=}1{,}996$ Movies validation is consistent ($-2.8\%$, $p{=}.091$); we make no claim about $d{<}0.10$.
\section{Implications and RAEP}\label{sec:implications}
\textbf{RQ5:} Can recall-aware evaluation diagnose when reranking is meaningful, and where should the field invest next (\S\ref{sec:raep})?

\subsection{What the Ceiling Means for Practice}\label{sec:practice}

At $r{=}0.08$ under LOO, $92\%$ of users have no relevant item in $C_u$ at all, and the remaining $8\%$ face a low-probability identification task on a set that already passed CF's prior. A cost audit sharpens the point: no LLM beats the \$0.00 CF baseline on Beauty, and Llama-3.3-70B (\$0.40/1K) is significantly worse on Movies ($p{=}.006$). \textbf{Takeaway.} Measure retrieval recall first; do not scale models under low recall (4B matches 671B here); A/B-test rather than trust offline results at $n{=}500$.

\subsection{The Recall-Aware Evaluation Protocol (RAEP)}\label{sec:raep}

We propose RAEP as a minimum methodological standard for future LLM-rerank work, with four components:
(\emph{1}) \textbf{Recall Diagnosis}, classifying the regime as CRITICAL ($r{<}5\%$), LOW ($5$--$15\%$), MODERATE ($15$--$30\%$), or ADEQUATE ($r{\geq}30\%$);
(\emph{2}) \textbf{RA-NDCG}, $\text{NDCG@}k(u)/\text{NDCG}^*_k(u)$, isolating reranker skill from retrieval;
(\emph{3}) \textbf{Ceiling utilisation $\eta$}, the fraction of available headroom a reranker exploits, scored against the reranker's own candidate window and reported only where the retrieved-positive subsample supports it (below $16\%$ on Beauty; too few users to estimate on Movies and Electronics);
(\emph{4}) \textbf{Adaptive-$K$}, per-user candidate sizing from the CF score distribution, $K{\in}[50,500]$.
Applied to our datasets, RAEP places Electronics ($r{=}2.24\%$) and Movies ($r{=}2.95\%$) in CRITICAL, Beauty ($r{=}8.23\%$) in LOW, and MovieLens-25M ($r{=}18.9\%$) in MODERATE. Adaptive-$K$ lifts Recall@$K$ by 69--112\% at 2--3$\times$ budget, but leaves NDCG@10 indistinguishable from fixed $K{=}100$ ($p\in\{.17,.54,\text{n/a}\}$): it is a diagnostic, not an NDCG-improvement pipeline. A retrospective audit of LLMRank~\cite{hou2024large} places its Amazon-Games ($K{=}20$, $N{=}16{,}859$) in CRITICAL and MovieLens-1M at the LOW edge, so gains measured there may not extrapolate to deployment-scale catalogs.

\subsection{The Only Escape Route Our Framework Permits}\label{sec:escape}

Theorem~\ref{thm:ceiling} is a statement about \emph{closed-candidate} rerankers: any reranker whose output is restricted to permutations of a retrieved set $C_u$. Three directions could in principle break the ceiling, but only one breaks the closed-candidate assumption itself. \emph{Better retrieval}~--- multi-source fusion at $K{=}3{,}000$ reaches 30\% recall on Beauty at $30\times$ candidate budget; domain-pretrained dense retrievers (E5-large, BGE-large with in-domain tuning) are a more promising lever~--- raises the ceiling but operates inside the same combinatorial constraint. \emph{Richer metadata}~--- on Movies a metadata-enrichment pilot at $n{=}200$ lifts NDCG@10 by 40\% but still sits 79\% below CF~--- improves $\eta$ but not the ceiling. \emph{Generative retrieval}~--- end-to-end token-level item generation (TIGER, GenRec)~--- is the only direction in our taxonomy that lets the model emit items absent from $C_u$, breaking the closed-candidate assumption of Theorem~\ref{thm:ceiling}. We therefore identify generative retrieval as the natural follow-up: a head-to-head against our hybrid CF$+$BM25$+$Dense baseline at the recall--NDCG operating point identified in this work.

\subsection{What the Oracle Protocol Is Still Good For}\label{sec:oracle_value}
Our critique is of a specific inference, not of the protocol itself.
Holding the candidate set fixed and guaranteeing the positive is present is a legitimate way to compare rerankers to \emph{each other}: it isolates ranking skill from retrieval variance, and it is the only setting in which prompt and model differences are measurable at all in our data (\S\ref{sec:prompts}).
What it cannot support is an absolute claim about deployment performance, because the quantity it holds fixed is exactly the one that binds in production.
Reporting an oracle number alongside the realistic number and the recall that separates them costs one column and removes the ambiguity.

\smallskip\noindent\textbf{Limitations and conclusion.}
We do not solve retrieval; the density model ($R^2{=}0.39$, $n{=}8$) is descriptive, and $n{=}500$ cannot detect $d{<}0.10$ effects.
Three boundaries deserve emphasis.
Our retrievers reach Recall@100 of 2--19\%; a production system operating at substantially higher recall sits outside the regime we measured, and the ceiling loosens as recall rises.
Our rerankers see CF embeddings, item text, popularity, and CF scores, but not the dwell-time, context, and cross-feature signals an industrial ranker consumes.
And offline evaluation cannot observe the online feedback loop through which a reranker's output changes future training data.
Within those boundaries, a theorem-backed bound, $\eta$ measurements on every tested reranker, and a five-dimension falsification point to one conclusion: under the retrieval we can achieve, no LLM reranking strategy we tested beats CF, and the binding constraint is retrieval rather than reranker optimisation.

\section*{GenAI Usage Disclosure}
An LLM coding assistant (Claude Code, Anthropic) was used to help write and refactor experiment and analysis code, and to copy-edit and condense the manuscript.
The LLMs evaluated in \S\ref{sec:results}--\S\ref{sec:falsification} are the object of study, not authoring tools, and are documented in \S\ref{sec:setup}.
All code, results, and claims were verified by the author, who takes full responsibility for the content.

\bibliographystyle{ACM-Reference-Format}
\balance
\bibliography{references}

\end{document}